\documentclass[letterpaper, 10pt, conference]{ieeeconf}

\IEEEoverridecommandlockouts
\usepackage{graphicx} % Required for inserting images
\usepackage{verbatim}
\usepackage{amsmath, amssymb}
\usepackage{graphicx}
\usepackage{babel}
\usepackage{color,soul}
\usepackage{svg}

\usepackage{color}

\newtheorem{theorem}{Theorem}
\newtheorem{lemma}{Lemma}
\newtheorem{remark}{Remark}

\newcommand{\cco}{{\rm co}}

\DeclareMathOperator{\interior}{int}

\title{Controllability Analysis of Multi-Modal Acoustic Particle Manipulation in One-Dimensional Standing Waves}
\author{Dongjun Wu$^{1}$, Guilherme Perticarari$^{2}$, Thierry Baasch$^{3}$
	\thanks{*This project has received funding from the European Research Council (ERC) under the European Union's Horizon 2020 research and innovation programme under grant agreement No 834142 (ScalableControl) and from the Swedish Research Council (No. 2022-04041)}
	\thanks{$^{1}$ D. Wu is with Department of Automatic Control, Lund
	University, Box 118, SE-221 00 Lund, Sweden {\tt\small
	dongjun.wu@control.lth.se}.}
	\thanks{$^{2}$ G. Perticarari is a student in the Master's in Machine Learning, Systems and Control programme at Lund University {\tt\small
	gj.perticarari@gmail.com}}
	\thanks{$^{3}$ T. Baasch is an Assistant Professor at the Department of Biomedical Engineering, Lund University, Box 118, SE-221 00 Lund, Sweden {\tt\small
	thierry.baasch@bme.lth.se}. }}

\begin{document}
\maketitle
\begin{abstract}
Acoustic manipulation in microfluidic devices enables contactless handling of biological cells for Lab-on-Chip applications. This paper analyzes the controllability of multi-particle systems in a one-dimensional acoustic standing wave system using multi-modal actuation. By modeling the system as a nonlinear control system, we analyze its global and local controllability, quantifying these properties in terms of mode numbers. Our results show that sufficient modes enable dense reachability sets, while mode mixing with 10 modes grants a strict notion of controllability to 80\% of the state space in a two-particle system. 
These findings offer theoretical insights for designing acoustic manipulation algorithms, supporting efficient control in biomedical applications.
\end{abstract}
{\keywords Biotechnology, Acoustic manipulation, Controllability analysis}

\section{Introduction}

Lab-on-Chip (LOC) is an active research field in biomedical technology aiming to miniaturize and automate applications. A promising category of LOCs relies on acoustic manipulation for the contactless handling of biological cells and micro-organisms within microfluidic devices. Applications of the technology include, for example, the separation of different types of Leukocyte sub-populations \cite{urbansky2019label} or circulating tumor cells from whole blood \cite{magnusson2024acoustic}. 
Typically, a single acoustic resonance mode is used for the manipulation of bulk biological material. Recently, a new technique based on multi-modal actuation has emerged. In fact, it has been shown experimentally that many resonance modes can be employed to gain control of individual particle paths, see e.g. \cite{zhou2016controlling, shaglwf2019acoustofluidic, schrage2023ultrasound, yiannacou2023acoustic} potentially enabling single-cell manipulation and optical trap-like functionality at much highly reduced cost.

Although several experimental works are presented in the literature, the corresponding theoretical investigations are limited. Möller \cite{moller2013acoustically} showed theoretically that continuous frequency sweeping allows pushing particles towards a chamber wall. Glynne-Jones et al. \cite{glynne2010mode} showed theoretically that mode-mixing between two modes creates a new stable equilibrium between the stable equilibria of the original modes. 

From a control-theoretic perspective, this system raises an interesting question regarding its controllability, i.e., which particle configurations can be achieved by a given number of modes. Answering this question is crucial for understanding which applications are viable with a reasonable number of modes of vibration. This paper addresses controllability for the motion of one or more particles in an ideal one-dimensional (1D) acoustic standing wave defined by a rigidly walled chamber. This process is depicted in Fig \ref{fig:schematic}, where a two-particle system changes its configuration due to a 1D standing wave being triggered from $t_1$ to $t_2$. Although one-dimensional fields are selected here for their relative simplicity, they are useful, and applications have been presented in the literature \cite{glynne2010mode, moller2013acoustically}. To the best of our knowledge, the only theoretical statement made in the literature regarding controllability and reachability is that the minimum number of modes required for full controllability must be at least equal to the dimension of the state space plus one \cite{zhou2016controlling, shaglwf2019acoustofluidic}. This statement follows trivially from our analysis in section \ref{Sec:locallyControl}. 

\begin{figure}[h]
    \centering
    \includegraphics[width=\linewidth]{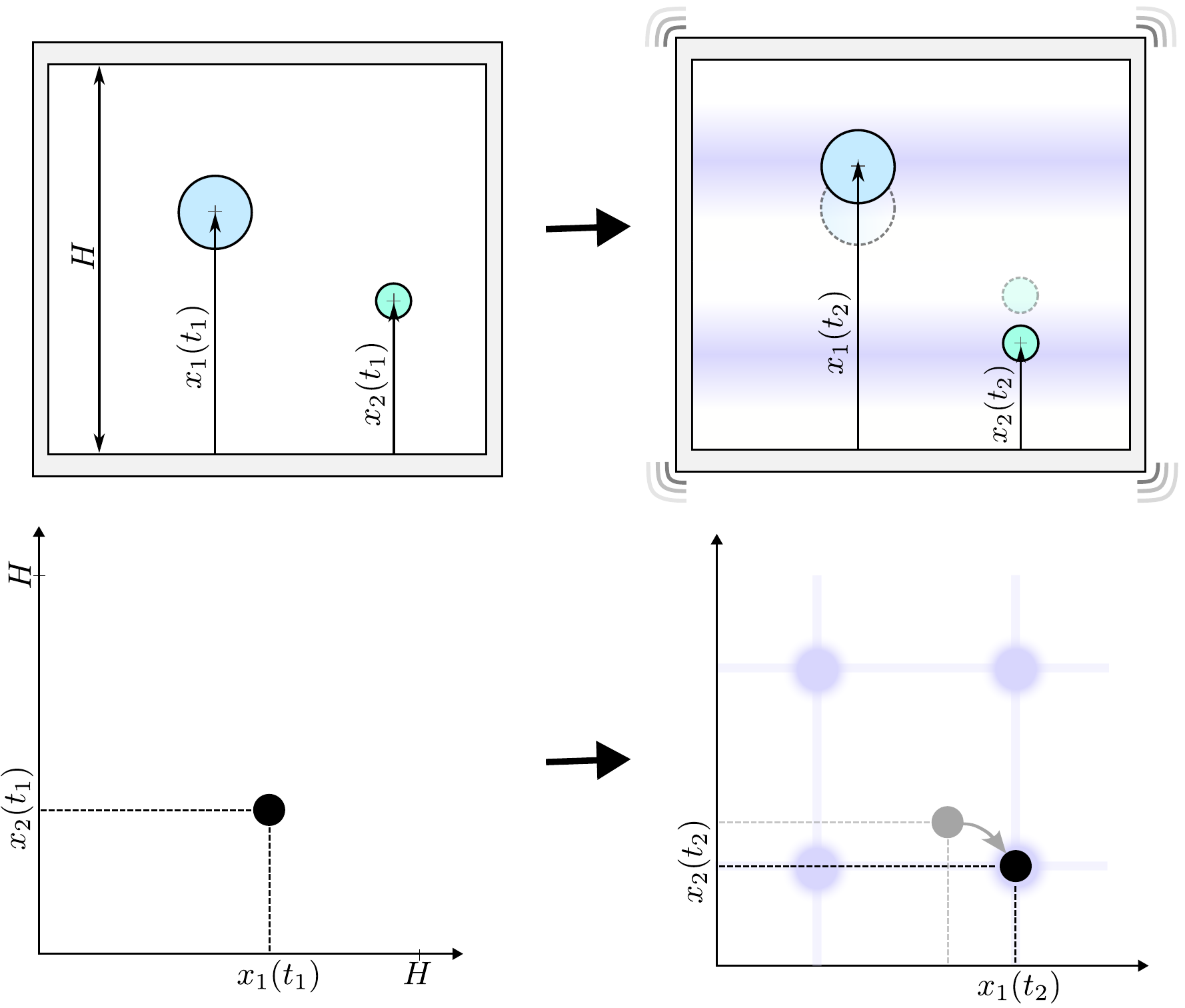}
    \caption{Schematic of the acoustic manipulation process. \textbf{(Top)} The system's configuration changes due to a 1D mode of vibration being played from $t_1$ to $t_2$. \textbf{(Bottom)} This process is shown from a state-space point-of-view.} 
    \label{fig:schematic}
\end{figure}

Our contributions are twofold:
\begin{itemize}
\item We establish global controllability by constructing a controllability graph for assignable stable equilibria, demonstrating that a sufficient number of modes ensures dense reachability sets within specified regions.
\item We analyze local controllability through mode mixing, identifying regions of locally controllable states via simulations, and quantify the impact of mode numbers on controllability in a two-particle system.
\end{itemize}

{\it Notations:} We collect some of the notations that will be used in the text. Given a subset $S\subseteq \mathbb{R}^n$, ${\rm co} (S)$ denotes the convex hull (the smallest convex set containing $S$) of $S$, and ${\rm int}(S)$ the set of interior points in $S$ (the largest open set in $S$). ${\rm \mathbf{N}} $ is the set of natural numbers. For a given natural number $n$, the set ${\rm I}_n$ denotes the set $\{1,2, \cdots , n\}$.

\begin{comment}
 DJ contributed to the writing, reviewing, and theoretical analysis. DJ Devised and implemented the approach in section \ref{Sec:ControllabilityAnalysis}. GP contributed to the writing, reviewing, and theoretical analysis. GP implemented the approach in section \ref{Sec:locallyControl}. 
 TB provided technical feedback and contributed to the writing, reviewing, and theoretical analysis.

\end{comment}

\section{Modeling and Problem Formulation} \label{model}
\subsection{One-dimensional (1D) acoustophoresis}
The acoustic radiation force acting on an isolated particle $i$ positioned at $x_i(t)$ of radius $a$ much smaller than the acoustic wavelength $\lambda$ is given by the negative gradient of the Gorkov potential \cite{gor1962forces, sapozhnikov2013radiation},
\begin{equation} \label{sys:grad:n_1d}
   F_{\mathrm{ac}}(x(t), y(t), z(t)) = - \nabla U(x(t), y(t), z(t)).
\end{equation}
In experiments, the geometry of the Gorkov potentials can be very complex and difficult to predict from first principle calculations. Therefore, in practice, measurements are required to build a reliable model of the system. If only 1D resonances are excited, such as shown experimentally in \cite{lamprecht2016imaging}, the situation is simpler, and some of the underlying phenomena can be modeled analytically to a reasonable approximation. In an ideal one-dimensional acoustic standing wave of mode $u \in \mathbf{N} $, the Gorkov potential for a particle $i \in [1, \, ... ,\, n]$ is given by 
%\begin{equation*}
%    U(x,u) = \frac{1}{2\pi} \sum_{i=1}^n  c_i \cos(2\pi u x_i)
%\end{equation*}

\begin{equation*}
    U(x_i(t),u) = \frac{3}{2} V_{i} E_{\mathrm{ac},u} \Phi_{i} \cos(2\pi \frac{u}{H} x_i(t)) + \text{const.},
\end{equation*}
where we introduced the particle volume $V$, contrast factor $\Phi$, the energy of the mode $E_{\mathrm{ac},u}$, and the channel height $H$, which defines our domain ($x_i(t) \in [0, H]$), see \cite{bruus2012acoustofluidics} for details. 

The motion of the particle is then given by balancing inertia $m \ddot{x}_i(t)$ with acoustic $F_{\mathrm{ac}}$ and hydrodynamic $F_{sk}$ forces. If the particle is isolated, i.e. sufficiently far from other particles or channel walls, the interactions can be neglected, and the equation of motion for a particle $i$ reduces to 

\begin{equation}
m \ddot{x}_i(t) = F_{\mathrm{ac}} (x_i(t)) + F_{\mathrm{sk}} (\dot{x}_i(t)),
\end{equation}
where $m$ denotes the particle's mass and $F_{\mathrm{sk}} (\dot{x}_i(t))= - 6 \pi a_i \eta \dot{x}_i(t)$ is the Stokes' drag experienced by the particle in a fluid of viscosity $\eta$.

It can be shown that the dynamic motion is, under normal circumstances, dominated by the viscous effects and inertia can be neglected when computing the particle trajectories \cite{baasch2017multibody}.

The equation of motion for a particle $i$ exposed to mode $u$ becomes $F_{\mathrm{sk}}(\dot{x}_i(t)) = - F_{\mathrm{ac}}(x_i(t)) = - \nabla_x U(x_{i}(t),u)$, or
% \begin{align}
% \dot{x}_i(t) &= - \frac{1}{4} \Phi_{i} \frac{a_i^2}{\eta}  \sum_u \alpha_u(t) E_{\mathrm{ac},u}  \frac{\partial }{\partial x} \left[  \cos \left( 2\pi \frac{u}{H} x_i(t) \right) \right], \\
% \dot{x}_i(t) &=  c_i  \sum_u \alpha_u(t) u E_{\mathrm{ac},u} \sin \left( 2\pi \frac{u}{H} x_i(t) \right) ,\, c_i= \frac{\pi}{2H} \Phi_{i} \frac{a_i^2}{\eta}, \\
% \end{align}
\begin{align}
    \dot{x}_i(t) = c_{i,u} u \sin\left( 2\pi  \frac{u}{H} x_i(t) \right)
    \label{eq:dynamics}
\end{align}with 
\begin{equation} \label{eq:x_i:phy}
    c_{i,u} = \frac{\pi a_i^2 \Phi_i E_{\mathrm{ac},u}}{2H\eta}
\end{equation}being a constant dependent on particle properties and mode number. 

% \blue{[Mode mixing is postponed]}

% Although energy $E_{\mathrm{ac},u}$ is in principle a function of the mode, assuming it to be independent of the mode has no impact on the controllability of the system.

% In this work, we investigate the underlying system and its controllability in general. 
%mathematical model of an $n$-particle 1d system can be described as
%\begin{equation} \label{sys:grad:n_1d}
%   F_{\mathrm{ac}} = - \nabla_x U(x,u) 
%\end{equation}where the mode-dependent Gorkov potential $U(x,u)$ is
%\begin{equation*}
%    U(x,u) = \frac{1}{2\pi} \sum_{i=1}^n  c_i\cos(2\pi u x_i)
%\end{equation*}
%and the input $u$ takes values in the set ${\rm I}_{N}=\{1,2,\cdots,N\}$, the state $x$ is restricted to the region $D_{n}:=[0,1]^{n}$. The coefficients $c_i$, are some given positive constants. 

The system is non-linear and multi-stable, in the sense that it has multiple stable and unstable equilibrium points. Fig. \ref{fig:pattern:k3} is an illustration of the potential function $U(x,3)$, namely, two particles under mode $u=3$. The red and blue marks denote stable and unstable equilibrium points, respectively. 
% Note that the pattern looks the same regardless of the values of $c_i$.

\begin{figure}[ht]
    \centering
    \includegraphics[scale=0.4]{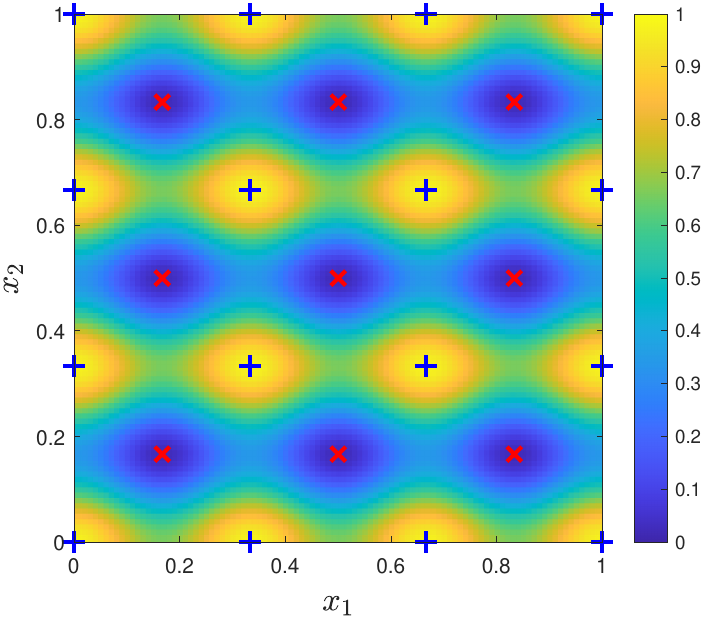}
    \caption{The acoustic force potential of a one-dimensional wave ($u=3$) is shown in the state space of two particles. Areas of high and low potential are shown in yellow and blue, respectively. The system trajectory will follow the negative gradient of the potential and stable equilibrium points are indicated by red crosses.}
    \label{fig:pattern:k3}
\end{figure}
 
% By computing the gradient of $U(x,u)$ explicitly, the system \eqref{sys:grad:n_1d} can also be written as

\subsection{Mathematical model for 1D multi-particle systems}
To facilitate mathematical analysis, we scale all $x_i$ in \eqref{eq:x_i:phy} by a factor $1 / H$ so that they lie within the unit interval $[0,1]$. In doing so, we can write a system with $n$ particles as:
\begin{equation} \label{sys:n_1d}
\left\{
\begin{aligned}
\dot{x}_{1} & =A_{1}u\sin(2\pi ux_{1})\\
 & \vdots\\
\dot{x}_{n} & =A_{n}u\sin(2\pi ux_{n})
\end{aligned}
\right.
\end{equation}with $x = [x_1,\cdots,x_n]^\top \in [0,1]^n$, $u\in \mathrm{I}_N$, and $A_i$ some $u$-dependent constants. The system of ordinary differential equations \eqref{sys:n_1d} will be called the {\it control system} or {\it system} for short, $x=[x_1, \cdots, x_n]^\top$ the {\it state} and $u$ the {\it control input}.

We have intentionally suppressed the dependence of $A_i$ on the mode $u$ (it would have been written as $A_{i,u}$ instead). In fact, for controllability analysis, it is safe to assume that $A_i$ are independent of $u$ thanks to the following fact:
\begin{lemma} \label{lem:Xu X1}
    The solution $X^u(t,\xi_0) \in [0,1]^n$ to the system \eqref{sys:n_1d} with mode $u \in \mathbb{N}$ and initial condition $\xi_0 \in [0,1]^n$ is fully characterized by the solution under mode $u=1$:
    \begin{equation}
        X^u_i(t,\xi_0) = \frac{1}{u}(p_i + X^1_i(q_i t, r_i))
    \end{equation}where $p=\lfloor u\xi_0 \rfloor$ and $r= u\xi_0 - p$ are the integer and fractional parts of the vector $u\xi_0$ respectively, and $q_i = \frac{u^2A_{i,u}}{A_{i,1}}$.
    % \begin{equation} \label{eq:Xu X1}
    %     X_u(t, \zeta_0) = 
    %     % \frac{A_{i,u}}{uA_{i,1}}
    %     \mathrm{diag}(A_{1,u}, \cdots, A_{n,u})
    %     X_1\left( {u} t , {u} \zeta_0 \right)
    % \end{equation}
\end{lemma}

\begin{remark}
Lemma \ref{lem:Xu X1} states that $X^u$ can be obtained by shifting and rescaling $X^1$ in addition to a time reparameterization. Equivalently, the integral curves under mode $u$ consist of a periodic repetition of a scaled version of the integral curves under the base frequency. Note that the constants $A_{i,u}$ only appear in the time reparameterization; thus, they do not affect the controllability of the system. 
However, by adjusting input power, one can speed up or slow down the movement of particles due to the fact that the energy $E_{{\rm ac},u}$ is proportional to input power.
For this reason, we assume $A_i$ are independent of the control input.
\end{remark}

We finish this subsection by mentioning a symmetry property of the base frequency ($u=1$) model, which will be useful for controllability analysis.

% The control input enters the system in a seemingly nonlinear way. However, the system trajectory is essentially characterized by the one under mode $u=1$, i.e., the base frequency model. Indeed, we have the following.
% \begin{proof}
% We prove by direct calculation:
% \begin{align*}
%     \frac{d}{d t}{X}_u(t,\zeta_0) & = \frac{1}{u} \frac{d}{d t} X_1(u t, u\zeta_0)
%      = u 
%     \begin{bmatrix}
%         c_1 \sin(2\pi X_1^{(1)} (u t, u \zeta_0)) \\
%         \vdots \\
%         c_n \sin(2\pi X_1^{(n)} (ut, u\zeta_0) )
%     \end{bmatrix} \\
%     & = u
%     \begin{bmatrix}
%         c_1 \sin(2\pi u X_u^{(1)} ( t,  \zeta_0)) \\
%         \vdots \\
%         c_n \sin(2\pi u X_u^{(n)} (t, \zeta_0))
%     \end{bmatrix} \\
% \end{align*}where the last equality is true by definition \eqref{eq:Xu X1}. The lemma now follows by invoking the uniqueness of the solution.
% \end{proof}

\begin{lemma} \label{lem:symm}
    The integral curves of the model are reflection-symmetric with respect to the hyperplanes $x_i = 1/2, \, \forall i \in {\rm I}_n$.
\end{lemma}

\begin{proof}
    Let $\alpha(t)$ and $\beta(t)$ be the solutions to $\dot{x}_i = A_i \sin(2\pi x_i)$ with initial conditions $\alpha(0)$ and $\beta(0)$ respectively.
    In view of
   \begin{align*}
        \frac{d(\alpha+\beta)}{dt}& = A_i (\sin(2\pi \alpha)
         + \sin(2\pi \beta )) \\
        & = 2A_i \sin\left( \pi (\alpha + \beta)\right) 
        \cos \left( \pi (\alpha - \beta) \right),
    \end{align*}
    we see that if $\alpha(0)+\beta(0) = 1$, meaning that the initial conditions are reflection symmetric w.r.t. $x_i = \frac{1}{2}$, then $\alpha(t)+\beta(t)=1$ for all $t\ge 0$. \footnote{For ordinary differential equation $\dot{x}=f(x,t)$, if $f(x_*,t) = 0$, $\forall t\ge 0$ for some $x_*$, then $x(t, x_*) = x_*, \, \forall t\ge 0.$ }

%     Let us consider the tangent vectors $\dot{x}_i(x_i)$ to the systems' trajectories for $x_i, \, \forall i \in {\rm I}_n$ given by $\dot{x}_i(x_i)=A_j \sin(2\pi u x_i)$. 
%     Reflection symmetry of the integral curves over a hyperplane $x_j=\frac{1}{2}$ implies that $\dot{x}_j(1/2+\alpha) = - \dot{x}_j(1/2-\alpha) $ for $\alpha \in [0,\, 1/2]$ while all other components of the tangent vector are required to remain constant under the transformation. The second condition is trivially true as all the components of the tangential vectors depend only on the corresponding state. Plugging the first condition into the equation of motion yields
%     \begin{align*}
%      \dot{x}_j \vert_{x_j=1/2+\alpha} &= A_j \sin(2\pi u (1/2 + \alpha) ), \\
%      &= A_j (-1)^u \sin(2\pi u \alpha ), \\
%      \dot{x}_j \vert_{x_j=1/2-\alpha} &= A_j \sin(2\pi u (1/2 - \alpha) ), \\
%      &= A_j (-1)^u \sin( - 2\pi u \alpha ), \\
%      &= - A_j (-1)^u \sin(2\pi u \alpha ), \\
%      &= - \dot{x}_j \vert_{ x_j=1/2+\alpha},
%      \end{align*}
% which concludes the proof.
   % Consider the dynamics $\dot{x}_i= A_i\sin(2\pi x_i) and coordinate transform $y = x + \frac{1}{2} $. It is easy to verify that $\dot{y}  = - c \sin(2\pi {y})$, which implies the desired symmetry.
   %  Therefore, 
   %  \begin{align*}
   %      \dot{\overbrace{x + \bar{x}}} & =  c(\sin(2\pi x) - \sin(2\pi \bar{x})) \\
   %      & = 2c \sin(2\pi(x+\bar{x}))\cos(2\pi(x-\bar{x}))
   %  \end{align*}From this we see that if.

\end{proof}

\section{Global Controllability Analysis} \label{sec:controllability}
\label{Sec:ControllabilityAnalysis}
After laying the foundations for modeling and problem formulation in Section \ref{model}, we are now in a position to study controllability.

\subsection{Assignable stable equilibria} \label{subsec:equilibrium}

Instead of characterizing controllability in the whole space $[0,1]^n$, we found it more tractable to study controllability on a discrete set of points called the assignable stable equilibria.
An assignable equilibrium $x^{*}=(x_{1}^{*},\cdots,x_{n}^{*})$ corresponds
to the solution of the equation $\sin(2\pi ux_{i}^{*})=0$, $\forall i\in{\rm I}_{n}$
and some $u\in{\rm I}_{N}$. We denote $E_{k}$ as the set of stable
equilibria when $u=k$, which can be explicitly computed as
\[
E_{k}=\left\{ \left(\frac{2i_{1}-1}{2k},\frac{2i_{2}-1}{2k},\cdots,\frac{2i_{n}-1}{2k}\right):i_{j}\in{\rm I}_{k},\;j\in{\rm I}_{n}\right\}.
\]
Thus, for mode $k$, there are $k^{n}$ stable equilibria. The set
of assignable stable equilibria is thus $E^{N}=\cup_{k=1}^{N}E_{k}$.
For convenience, let us write 
\begin{equation}
E_{k}(i_{1},\cdots,i_{n})=\left(\frac{2i_{1}-1}{2k},\frac{2i_{2}-1}{2k},\cdots,\frac{2i_{n}-1}{2k}\right).
\end{equation}
The region of attraction (ROA) for equilibrium $E_{k}(i_{1},\cdots,i_{n})$ under mode $k$ is the open cube 
\begin{equation}
\begin{aligned}
&\; R_{k}(i_{1},  \cdots,i_{n}) \\
& =\left]\frac{i_{1}-1}{k},\frac{i_{1}}{k}\right[\times\left]\frac{i_{2}-1}{k},\frac{i_{2}}{k}\right[\times\cdots\times\left]\frac{i_{n}-1}{k},\frac{i_{n}}{k}\right[.
\end{aligned}
\end{equation}
If $q\in E^{N}$, denote $R_{k}(q)$ the ROA of $q$ under mode
$k$.

\subsection{The controllability graph}

Given $p\in [0,1]^n$, and an equilibrium $q\in E^{N}$, if there exists
a mode $k\in{\rm I}_{N}$, such that $p\in R_{k}(q)$, then by applying mode
$k$, $p$ will converge to the equilibrium $q$ asymptotically.

Based on this observation, we can build a directed graph, denoted as $G_N$, to describe the reachability between different stable equilibria. The vertex set of $G_N$ is $E^{N}$ and,
given two equilibria $p,q\in E^{N}$, if there exists a mode $k\in \mathrm{I}_N$ that
renders $p$ to $q$ asymptotically, we say that $(p,q)$
is an edge of $G_{N}$. 
%\subsubsection*{Pattern switching between steady states}
%The first control strategy is to switch between stable equilibrium points.
As shown in Fig. \ref{fig:steady_switch}, the two stable assignable equilibria
$E_3(1,3)$ and $E_2(1,2)$ lie in the intersection of $R_3(1,3)$ and $R_2(1,2)$. Therefore,
$E_3(1,3)$ can be rendered to $E_2(1,2)$ by applying $u=2$ and vice versa. As a
consequence, there are two edges with different directions between $E_3(1,3)$ and $E_2(1,2)$, see Fig. \ref{fig:steady_switch}. 
%We refer to state transitions using this strategy as {\it Type 1 control.} 
We call $G_N$ the {\it controllability graph}.

\begin{figure}[ht]
    \centering
    \includegraphics[scale=0.35]{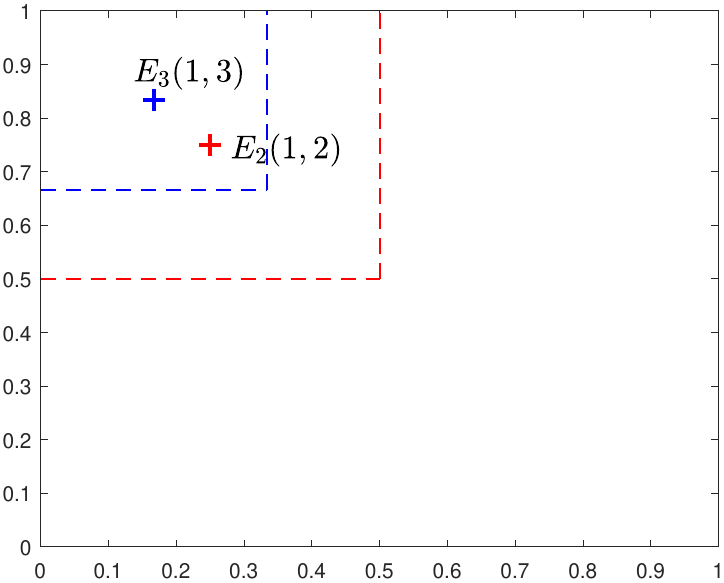}
    %\hspace{1em}
    \includegraphics[scale=0.4]{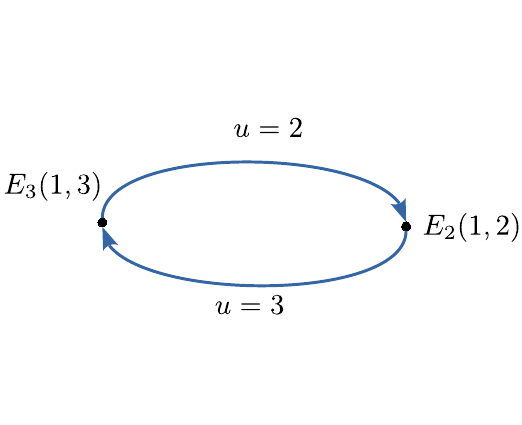}
    \caption{$E_3(1,3)$ and $E_2(1,2)$ are reachable from each other by switching between $u=3$ and $u=2$.}
    \label{fig:steady_switch}
\end{figure}

Note that the values of the coefficients $A_i$ do not play a role in the above analysis, but they do affect the controllability properties near the diagonal line.
%\subsubsection*{The diagonal line}
%The reachability of the system is critically
%related to the values of the coefficients $c_{i}$.
Indeed, when $A_{1}=A_{2}=\cdots=A_{N}$, the diagonal line 
$\{x_1 = x_2 = \cdots = x_n \}$ is an invariant
set, implying that it is impossible to leave this set. 
This corresponds to the fact that
it is impossible to separate identical particles if they have the same initial conditions. Better controllability is enabled when $A_{i}$'s become different.

\begin{figure}[ht]
\centering
\includegraphics[scale=0.5]{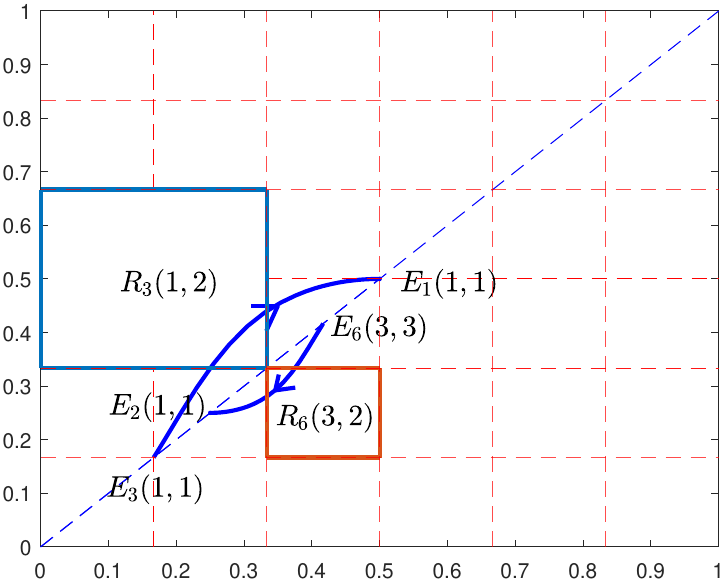}
% \includegraphics[scale=0.5]{pic/diagN3-cropped}
% \vspace{2em}
% \includegraphics[scale=0.5]{pic/diagN6-cropped}
% \par
\caption{
$u=1$ drives $E_3(1,1)$ to $E_1(1,1)$ by passing through $R_3(1,2)$. $u=2$ drives $E_6(3,3)$ to $E_2(1,1)$ by passing through $R_6(3,2)$.
\label{fig:u1,N3}}
\end{figure}

Fig. \ref{fig:u1,N3} plots (in blue) two trajectories that exit the diagonal line, showing it is possible to leave and enter the diagonal line, and hence allows adding extra links to the controllability graph. 
For example, when the system is at initial state $E_3(1,1)$,
by applying mode $u=1$, the state will be rendered to the
equilibrium $E_{1}(1,1)$ asymptotically, while passing
through $R_{3}(1,2)$ at a certain time $T$. By switching to mode $3$ at time $T$, the
state will be rendered to $E_3(1,2)$ instead. 
%We refer two this strategy as {\it Type 2
%control}.
Therefore, node $E_{3}(1,2)$ can be connected 
by an edge from $E_{3}(1,1)$. Similarly, by symmetry, $E_{6}(3,3)$ connects
to $E_{6}(3,2)$. 

Finally, we remark that the lines $x_i = 1/2$ are always invariant regardless of the values of $A_i$. Therefore, we shall restrict controllability discussions within regions such as $0<x_i<1/2$.

\subsection{Controllability graph of a two-particle system}
We are now ready to construct the controllability graph of a two-particle system -- a system
that is simple but illustrative enough to assess the controllability of the device.

Two steps are needed to construct the graph. In Step 1, we draw an edge from each node to every other node that is reachable under a single control input $u\in \mathrm{I}_N$.
In Step 2, the dynamics on the diagonal lines are taken into consideration, so that links
from the diagonal line to non-diagonal states are established.

Once the graph has been obtained, we decompose the graph into {\it strongly connected components} (SCC). By SCC, we mean that each point in the component is reachable asymptotically from every other node under a sequence of control inputs. The decomposition can be done using standard algorithms, e.g., Tarjan's algorithm.

Fig. \ref{fig:noDiag} shows the SCCs for $N=6, \, 8, \, 9,  \, 12$ without considering the ``diagonal effects''. 
The figure suggests that as the number of modes, i.e., $N$ increases, the number of disconnected components decreases. In particular, when $N=9$, all assignable stable equilibria within the region $\{ 0<x_1<1/2, \; x_2 < x_1 \}$ are reachable from each other. The figure also confirms the symmetry described in Lemma \ref{lem:symm}.

\begin{figure}[ht]
    \centering
    \includegraphics[scale=0.3]{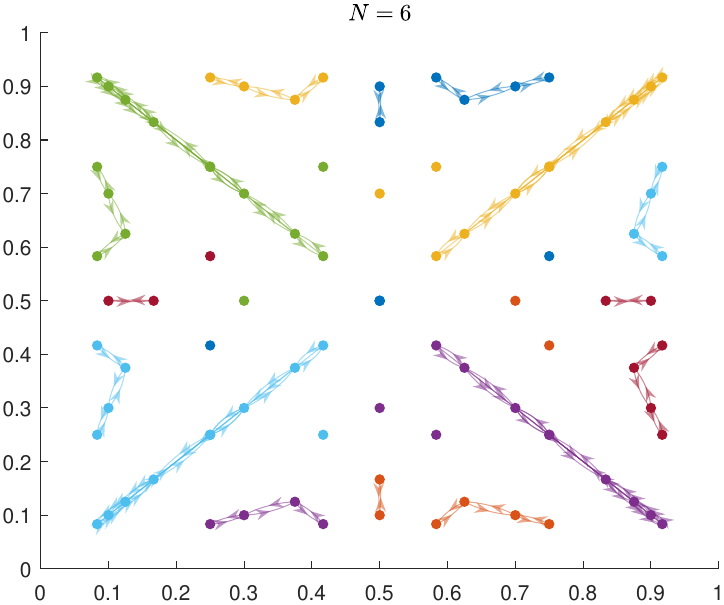}
    \includegraphics[scale=0.3, page=2]{pic/noDiagN589-cropped.pdf}
    \includegraphics[scale=0.3, page=3]{pic/noDiagN589-cropped.pdf}
    \includegraphics[scale=0.3]{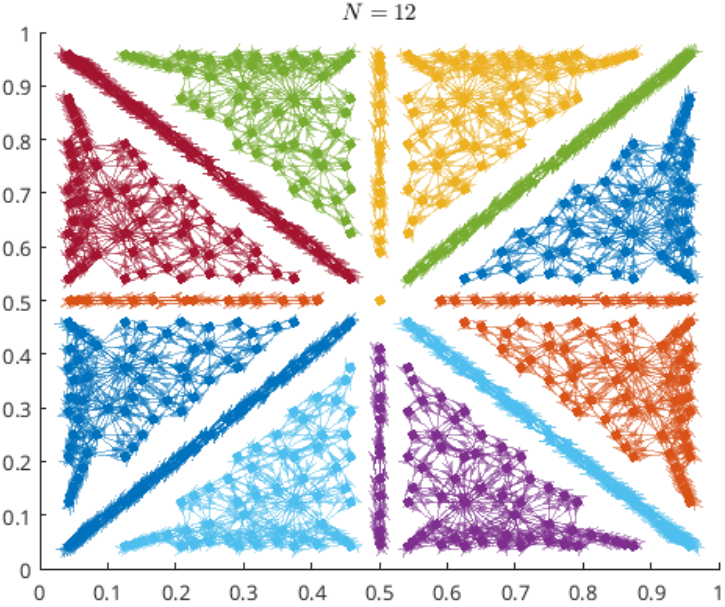}
    \caption{Strongly connected components of the reachability graph obtained by pattern switching between steady states. Available modes from left to right: $N=6, \, 8, \, 9, \, 12$.}
    \label{fig:noDiag}
\end{figure}

% Next, we add the diagonal effects, see Fig. \ref{fig:with-diag N9}. From Fig. \ref{fig:with-diag N9}, we see that all the assignable stable equilibrium points within $]0, 1/2[ \times ]0, 1/2 [$, or $R_2(1,1)$ using notation introduced in Section \ref{subsec:equilibrium}, with $N=9$ are reachable from each other. Due to symmetry, we conclude that the graph with $N=9$ is partitioned into $9$ strongly connected components. 

Next, by adding links emitting and entering the diagonal line, better controllability properties are gained. Fig. \ref{fig:with-diag N9} shows the SCCs of the graph in $]0, 1/2[$ obtained by adding the four extra links shown in Fig. \ref{fig:u1,N3}, i.e., $E_3(1,1) 
\to E_3(1,2)$, $E_3(1,2) \to E_1(1,1)$, $E_6(3,3) \to E_6(3,2)$, $E_6(3,2) \to E_2(1,1)$. Note that, by only adding four extra edges, the graph becomes strongly connected. 

\begin{figure}[ht]
    \centering
    \includegraphics[scale=0.3]{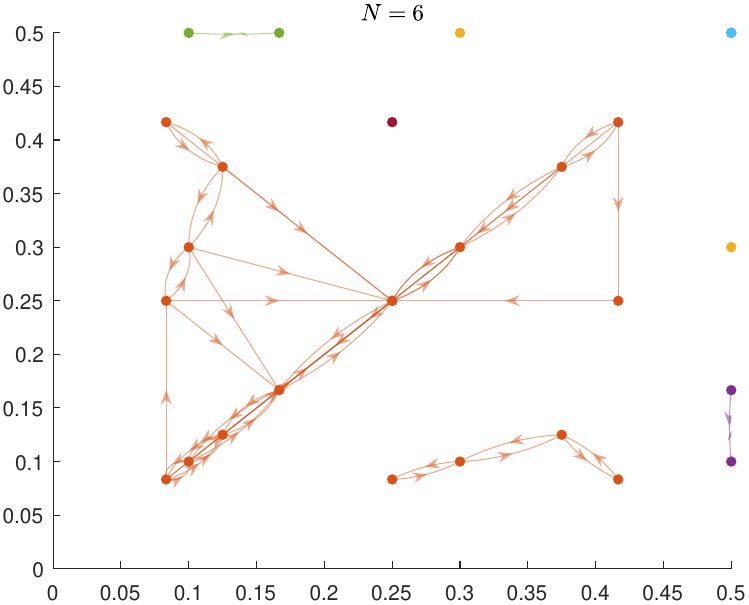}
    \includegraphics[scale=0.3]{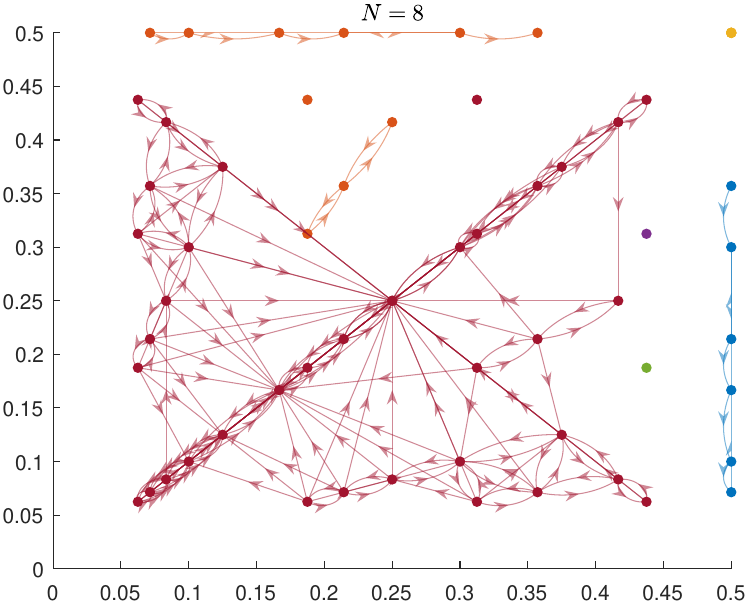}
    \includegraphics[scale=0.3, page=1]{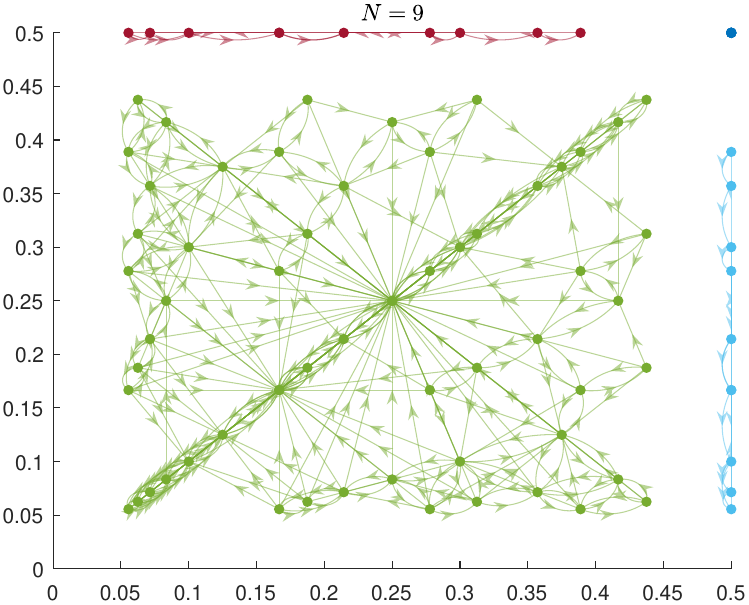}
    \includegraphics[scale=0.3, page=2]{pic/withDiag-cropped.pdf}
    \caption{Diagonal effects added. From left to right: $N=6, 8, 9, 12$.}
    \label{fig:with-diag N9}
\end{figure}

Thanks to Lemma \ref{lem:Xu X1}, we have the following theorem.
\begin{theorem}\label{thm:2part}
    If $A_1 \ne A_2$, then as $N \to \infty$, the reachability sets within $R_2(i,j)$, for $i,j=1,2$, are dense subsets. 
\end{theorem}
\begin{proof}
    It is sufficient to consider $R_2(1,1)$. Partition $R_2(1,1)$ into four pieces: $R_4(i,j)$, $i,j \in {\rm I}_2$. In each of the four regions, one can use controls from ${\rm I}_{18}$ to form a strongly connected component of the assignable stable equilibrium points. After that, the four components can be connected using controls from ${\rm I}_9$. Repeating this procedure, we obtain a dense strongly connected component in $R_2(1,1)$.
    % Then by Lemma \ref{lem:Xu X1}, the 
\end{proof}

\begin{remark}
    Theorem \ref{thm:2part} can be extended to arbitrarily many particles easily. We report this in a coming work.
\end{remark}

% \red{Plot some typical paths}

\section{Local Controllability Analysis}

In this section, we study another notion of controllability, namely, {\it local controllability}. Unlike global controllability in Section \ref{Sec:ControllabilityAnalysis}, which is studied over assignable stable equilibria, we study local controllability in the whole space $[0,1]^n$. For simplicity and visualization purposes, we restrict to $n=2$.

Our study of local controllability is based on mode mixing, a
concept that is closely related to relaxations of nonlinear systems. 

\subsection{Relaxation by Mode Mixing}
% A relaxation of the system \eqref{sys:n_1d} can be posed as following.

% \blue{[moved to here]}
% where we summed the individual components of multiple simultaneously actuated modes. The functions $\alpha_u(t) \geq 0$ define the control scheme. Physically, the $\alpha_u(t)$ denote the input of energy into the respective modes as a function of time.
% The aim is to find $\alpha_u(t)$ to gain control of the individual trajectories of the particles $x_i(t)$. Two cases can be distinguished, the standard mode-switching case and the mode-mixing case. For the mode switching case, the total time is decomposed into $J$ (non-overlapping) time intervals $t_j \in [t_j^S, t_j^E], \, j \in [0, J] $ for which the control input is constant. The mode mixing case allows multiple modes to be applied simultaneously, as long as the total input energy is limited. Formally, we have

% \begin{align}
% &\alpha_u(t_j)=1, \text{ for } u=u_j \text{, } \alpha_u(t_j)=0 \text{ for } u \neq u_j \text{ for mode-switching and }  \\
% & \sum_u \alpha_u(t) \leq 1\text{ for mode-mixing.} 
% \end{align}
For the system \eqref{sys:n_1d}, when the available input modes are fixed, say $u \in {\rm I}_N$, then the system \eqref{sys:n_1d} can be expressed as a differential inclusion
\begin{equation} \label{sys:diff_inc:n_1d}
    \dot{x} \in F(x),
\end{equation}
where \( F(x) = \{ f_1(x), \cdots, f_N (x) \} \) with 
\begin{equation}
    f_u (x) = 
    \begin{bmatrix}
        A_1 u \sin(2\pi u x_1) \\
        \vdots \\
        A_n u \sin(2\pi u x_n)
    \end{bmatrix}
\end{equation}

Mode mixing amounts to relaxing the system \eqref{sys:diff_inc:n_1d} by 
\begin{equation} \label{sys:relax:n_1d}
    \dot {x} \in \cco (F(x))
\end{equation}where $\cco (F(x))$ is the convex hull spanned by $F(x)$.

The celebrated Filippov-Ważewski theorem asserts that the reachability region of the system \eqref{sys:relax:n_1d} is dense in the reachability set of the system \eqref{sys:diff_inc:n_1d}.
\begin{theorem}[Filippov-Ważewski \cite{aubinDifferentialInclusionsSetvalued1984}]
Suppose that $f_j$ are Lipschitz for $j\in {\rm I}_N$. Let $x(t)$ be a bounded solution to \eqref{sys:diff_inc:n_1d} on $[0,T]$, then for any $\epsilon>0$, there exists a solution $y(t)$ to \eqref{sys:relax:n_1d} with $y(0)=x(0)$, such that $|x(t) - y(t)|<\epsilon$ for all $t\in ]0,T]$.
\end{theorem}

In view of the Filippov-Ważewski theorem, it is sufficient in practice to study the reachability and controllability properties of the relaxed/covexified system, which is much better behaved, as the integer-type constraints are now replaced by convex constraints. More precisely, we can write the system \eqref{sys:relax:n_1d} as
\begin{equation} \label{sys:convex:n_1d}
    \dot{x} = \sum_{j=1}^N w_j(t) f_j(x)
\end{equation}where $w_j(t)$ are the control inputs which take values in the $N$-simplex $\Delta_N = \{ w\in \mathbb{R}^N_{\ge 0} : \sum_{j=1}^N w_j = 1, \, w_j \ge 0, \, \forall j \in {\rm I}_N \}$.

We call the system \eqref{sys:convex:n_1d} ``mode mixing'' in that the instantaneous control input is a mixture of different available modes. To recover a true control input, i.e., the sequence of mode switching, which approximates the system dynamics well, we can use fast switching between the available modes. 

\subsection{Regions of locally controllable states}
\label{Sec:locallyControl}

% -state => locally controllable state // (ok!)
% -region => the system is controllable in this region

We rely on (\ref{eq:dynamics}) to simulate $\dot{x}$ under mode mixture $w$ in a two-particle system with given physical properties (e.g., particle radius). Using this simulation, we hope to find simply connected regions in the state-space consisting solely of states that can freely move in any direction around themselves. 
% A state with that property will be named \textit{locally controllable}. 
These regions are governed by a strict notion of controllability, since all of its states are guaranteed to reach each other by any path contained in them.

{A state $x$ is called {\it locally controllable} if it lies in} the interior of the convex hull of $F(x)$, i.e. \footnote{We adopt the convention that an element in $F(x)$ is a tangent vector that ``sits'' at $x$, and that $\cco (F(x))$ is a convex set that ``sits'' at $x$. \eqref{eq:loc-ctrl} would have been replaced by $0 \in {\rm int}\, \cco (F(x))$ if the elements in $ F(x)$ are seen as vectors sitting at the origin.}
\begin{equation} \label{eq:loc-ctrl}
x \in \interior \cco (F(x)).
\end{equation}
% which means that $\dot{x}$ can assume values in any direction.
{The naming is justified by the fact that the property of local controllability is commonly applied to a state x if, at x, the system is able to move instantaneously in every possible direction in the state-space.}
{Note that it is part of the claim in \eqref{eq:loc-ctrl} that $\cco (F(x))$ has non-empty interior.}
% Note that since $f_0(x)=\bold{0}$ is in $F(x)$, we have that $\bold{0} \in \cco (F(x))$ always holds.

The number of modes, i.e.,  $N$, has a direct impact on the local controllability of a state, since it directly changes set $F(x)$. An example of how $N$ affects the local controllability of a state $x$ is given in Fig. \ref{fig:example_cvxhull}, where state $x$ only becomes locally controllable with the addition of one extra mode.

\begin{figure}
    \centering
    \includegraphics[width=\linewidth]{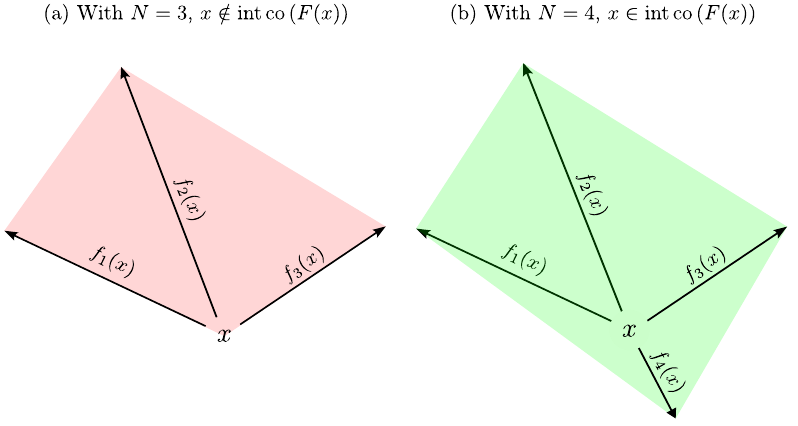}
    \caption{Example of how $N$ can influence the local controllability of a state $x$. \textbf{(a)} With $N=3$, $x$ is not locally controllable due to directional blind-spots around it. \textbf{(b)} The addition of $f_4(x)$ makes $x$ locally controllable.}
    \label{fig:example_cvxhull}
\end{figure}

Verifying the local controllability of individual states with our simulation is straightforward since the operation only requires a convex hull inclusion test. Discovering entire simply connected regions of these states, on the other hand, is a more daunting task, because we are working with a continuous state-space. Therefore, we only tested states from a thinly spaced grid that was overlaid on the space, and if a state was deemed locally controllable, we considered this property to be applied to the entire cell containing it, thus making it a simply connected region of locally controllable states. This assumption is only true if the grid cells are small enough that linearization can be sensibly applied inside them. For this process, we simulated the acoustic manipulation process with two particles with radii $a_1 = 1$ µm and $a_2 = 2$ µm dispersed in a 2-dimensional rectangular device with channel height $H=800$ µm. Fig. \ref{fig:grid} shows the results for $N=5$ and a grid spacing of $5$ µm, 2.5 times the size of the bigger particle, resulting in 25,281 points in the state-space. We found that 58.4\% of these points were locally controllable and that at least $N=10$ modes were necessary for this number to reach 80\%. Note how states $x$ in the symmetry lines are never deemed locally controllable due to the co-linearity of vectors $f_j(x)$, which makes $\cco \left( F(x) \right)$ devoid of an interior.

\begin{figure}
    \centering
    \includegraphics[width=0.6\linewidth]{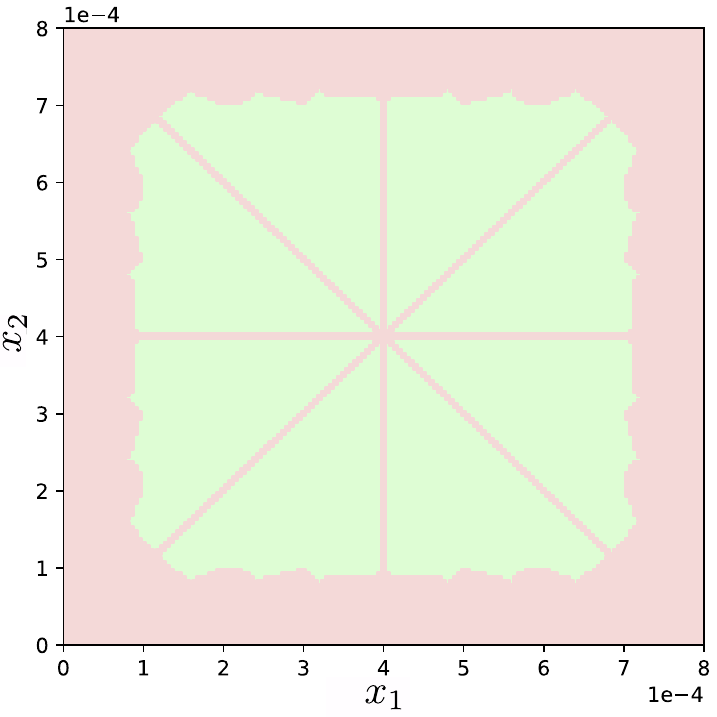}
    \caption{Approximation of simply connected regions of locally controllable states after the state space is discretized into a 159-by-159 grid. Locally controllable states are shown in green, while the remaining ones are shown in red. With $N=5$, 58.4\% of the states in the grid are locally controllable.}
    \label{fig:grid}
\end{figure}

So far, we restricted our investigation to two-particle systems, but real-life applications will generally require controlling more particles simultaneously. To understand how many modes our system would require to perform one-dimensional control for a set of $p$ particles, we repeated the grid approximation experiment with $p$ ranging from 2 to 10 and $N$ ranging from 2 to 20. Due to the curse of dimensionality, it is naturally unfeasible to verify the local controllability of all states in a $p$-dimensional grid if $p$ is too great and, therefore, we only tested a predefined number of states sampled from the grid for each $p$. The percentage of locally controllable states out of all samples was then used as a proxy for how controllable a system of $p$ particles and $N$ modes is. We used 3,000 samples for all $p$, and, for each new particle added to the system, we sampled its radius from a uniform distribution from $1$ µm to $2$ µm, with the remaining properties remaining the same as in the two-particle test (Fig. \ref{fig:grid}). Note that, from the definition of local controllability, a system with $p$ particles can only have locally controllable states if $N \geq p +1$, which is in line with refs \cite{zhou2016controlling,shaglwf2019acoustofluidic}. The result of these experiments (Fig. \ref{fig:multiple_part}) shows that, although the addition of modes for each $p$ increases the percentage of locally controllable states, such effect diminishes aggressively as $N$ becomes large. Moreover, the higher $p$ is, the weaker this effect is. Indeed, while $N=3$ modes yield a percentage of $\approx 30\%$ for a two-particle system, a ten-particle system remains shy of $~30\%$ even with $N=20$ modes.

\begin{figure}
    \centering
    \includegraphics[width=\linewidth]{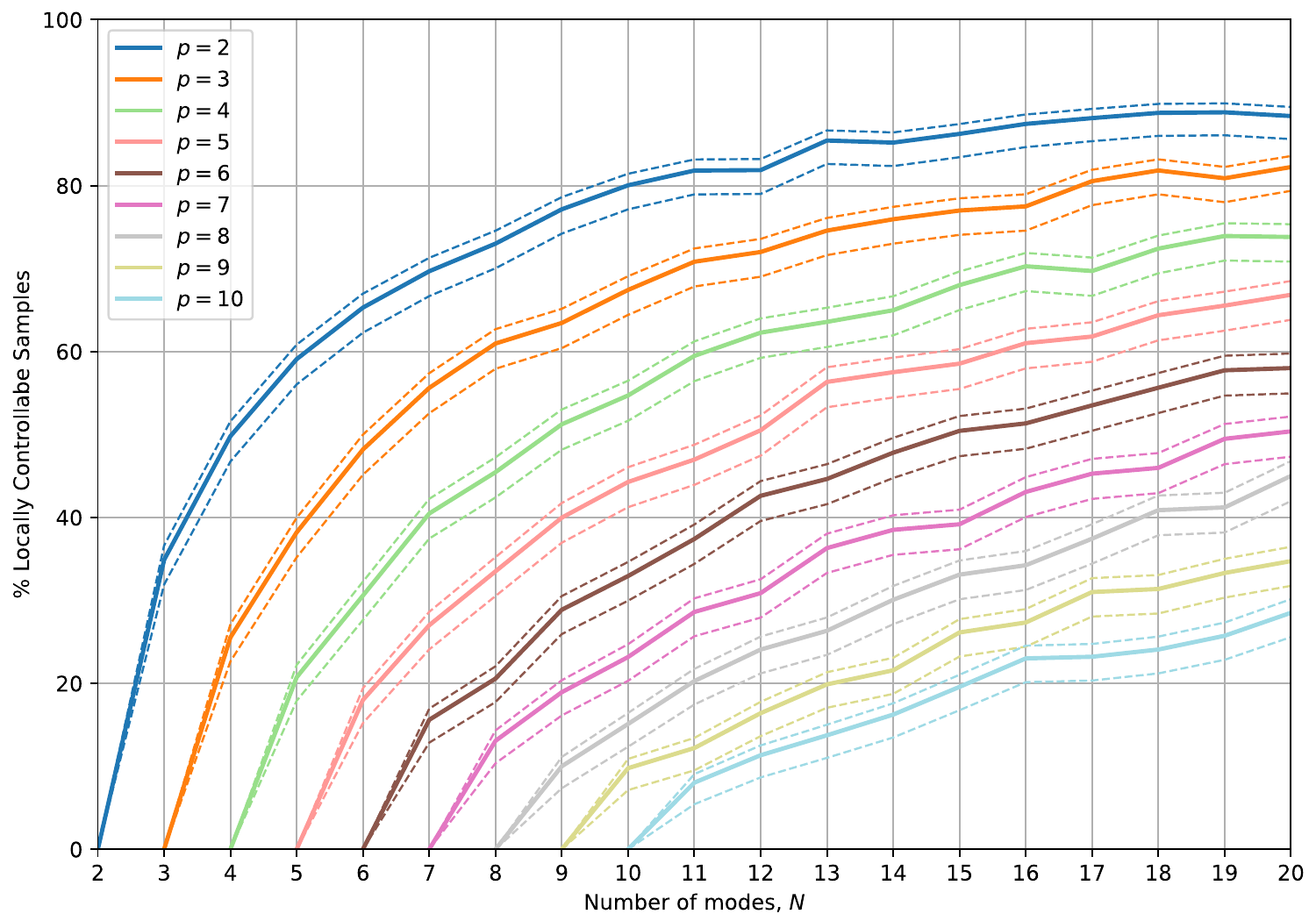}
    \caption{Plot showing the relationship between the percentage of locally controllable states and the number of modes, $N$, for each number of particles $p$. Each experiment is run on 3,000 states sampled from a $p$-dimensional grid. The solid lines denote the percentage of locally controllable states, while the dashed lines delimit the 95\% confidence interval estimate using the Wilson Score Interval calculation.}
    \label{fig:multiple_part}
\end{figure}

\section{Conclusion}

This study provides some theoretical controllability analysis of a one-dimensional acoustic manipulation device using multi-modal actuation. 
% For simplicity and visualization reasons, mainly two-particle systems are analyzed.
By modelling the system as a nonlinear control system, global and local controllability are rigorously defined. Globally, we construct a controllability graph to quantify global controllability among stable assignable equilibria, showing that a sufficient number of modes ensures dense reachability sets.
Locally, by employing mode mixing, we identify regions of controllable states, with simulations indicating that 10 modes achieve 80\% reachability in a two-particle system. These findings offer a theoretical foundation for acoustic manipulation for Lab-on-Chip applications. Future work will extend the analysis to multi-particle systems and experimental validation.
% 
% \red{We prefer adding funding info as footnotes on the first page.}
% 

%\bibliographystyle{IEEEtran}
\bibliographystyle{plain}
\bibliography{acoustic}

\begin{thebibliography}{10}

\bibitem{aubinDifferentialInclusionsSetvalued1984}
Jean-Pierre Aubin and Arrigo Cellina.
\newblock {\em Differential Inclusions: Set-Valued Maps and Viability Theory}.
\newblock Number 264 in Grundlehren Der Mathematischen {{Wissenschaften}}.
  Springer, Berlin New York, 1984.

\bibitem{baasch2017multibody}
Thierry Baasch, Ivo Leibacher, and J{\"u}rg Dual.
\newblock Multibody dynamics in acoustophoresis.
\newblock {\em The Journal of the Acoustical Society of America},
  141(3):1664--1674, 2017.

\bibitem{bruus2012acoustofluidics}
Henrik Bruus.
\newblock Acoustofluidics 7: The acoustic radiation force on small particles.
\newblock {\em Lab on a Chip}, 12(6):1014--1021, 2012.

\bibitem{glynne2010mode}
Peter Glynne-Jones, Rosemary~J Boltryk, Nicholas~R Harris, Andy~WJ Cranny, and
  Martyn Hill.
\newblock Mode-switching: A new technique for electronically varying the
  agglomeration position in an acoustic particle manipulator.
\newblock {\em Ultrasonics}, 50(1):68--75, 2010.

\bibitem{gor1962forces}
Lev~Petrovich Gorkov.
\newblock On the forces acting on a small particle in an acoustical field in an
  ideal fluid.
\newblock In {\em Sov. Phys.-Doklady}, volume~6, pages 773--775, 1962.

\bibitem{lamprecht2016imaging}
Andreas Lamprecht, Stefan Lak{\"a}mper, Thierry Baasch, Iwan~AT Schaap, and
  Jurg Dual.
\newblock Imaging the position-dependent 3d force on microbeads subjected to
  acoustic radiation forces and streaming.
\newblock {\em Lab on a Chip}, 16(14):2682--2693, 2016.

\bibitem{magnusson2024acoustic}
Cecilia Magnusson, Per Augustsson, Eva Undvall~Anand, Andreas Lenshof, Andreas
  Josefsson, Karin Wel{\'e}n, Anders Bjartell, Yvonne Ceder, Hans Lilja, and
  Thomas Laurell.
\newblock Acoustic enrichment of heterogeneous circulating tumor cells and
  clusters from metastatic prostate cancer patients.
\newblock {\em Analytical Chemistry}, 96(18):6914--6921, 2024.

\bibitem{moller2013acoustically}
Dirk~Bj{\"o}rn M{\"o}ller.
\newblock {\em Acoustically driven particle transport in fluid chambers}.
\newblock PhD thesis, ETH Zurich, 2013.

\bibitem{sapozhnikov2013radiation}
Oleg~A Sapozhnikov and Michael~R Bailey.
\newblock Radiation force of an arbitrary acoustic beam on an elastic sphere in
  a fluid.
\newblock {\em The Journal of the Acoustical Society of America},
  133(2):661--676, 2013.

\bibitem{schrage2023ultrasound}
Matthijs Schrage, Mahmoud Medany, and Daniel Ahmed.
\newblock Ultrasound microrobots with reinforcement learning.
\newblock {\em Advanced Materials Technologies}, 8(10):2201702, 2023.

\bibitem{shaglwf2019acoustofluidic}
Zaid Shaglwf, Bjorn Hammarstr{\"o}m, Dina Shona~Laila, Martyn Hill, and Peter
  Glynne-Jones.
\newblock Acoustofluidic particle steering.
\newblock {\em The Journal of the Acoustical Society of America},
  145(2):945--955, 2019.

\bibitem{urbansky2019label}
Anke Urbansky, Franziska Olm, Stefan Scheding, Thomas Laurell, and Andreas
  Lenshof.
\newblock Label-free separation of leukocyte subpopulations using high
  throughput multiplex acoustophoresis.
\newblock {\em Lab on a Chip}, 19(8):1406--1416, 2019.

\bibitem{yiannacou2023acoustic}
Kyriacos Yiannacou and Veikko Sariola.
\newblock Acoustic manipulation of particles in microfluidic chips with an
  adaptive controller that models acoustic fields.
\newblock {\em Advanced Intelligent Systems}, 5(9):2300058, 2023.

\bibitem{zhou2016controlling}
Quan Zhou, Veikko Sariola, Kourosh Latifi, and Ville Liimatainen.
\newblock Controlling the motion of multiple objects on a chladni plate.
\newblock {\em Nature communications}, 7(1):12764, 2016.

\end{thebibliography}

\end{document}